\documentclass[submission,copyright,creativecommons]{eptcs}
\providecommand{\event}{AREA 2023} 

\usepackage{iftex}

\ifpdf
  \usepackage{underscore}         
  \usepackage[T1]{fontenc}        
\else
  \usepackage{breakurl}           
\fi

\usepackage{amsfonts}
\usepackage{amsmath}
\usepackage{amssymb}
\usepackage{amsthm}
\usepackage{stmaryrd}
\usepackage{graphicx} 
\usepackage{adjustbox}
\usepackage{proofzilla}
\newtheorem{prop}{Proposition}
\newtheorem{defi}{Definition}

\newtheorem{lemma}{Lemma}
\def\myparagraph#1{\noindent\textbf{#1}. }

\usepackage{mathrsfs}

\usepackage{marginnote}

\def\defn#1{\textbf{#1}}

\def\next{\mkern1mu\mathord{\mathsf{X}}\mkern1mu}
\def\until{\mkern1mu\mathord{\mathsf{U}}\mkern1mu}
\def\release{\mkern1mu\mathord{\mathsf{R}}\mkern1mu}
\def\imp{\to}

\def\defn#1{\textbf{#1}}

\def\sem#1{\llbracket #1\rrbracket}

\def\sse{\leftrightarrow}

\def\M{\mathfrak M}
\def\N{\mathfrak R}     

\def\Ap{\mathcal{P}}
\def\set#1{\{#1\}}
\def\tuple#1{\langle #1 \rangle}

\def\some{\mathsf{E}\,}
\def\all{\mathsf{A}\,}

\newemptyvertex{node}{draw,circle}
\newemptyvertex{label}{}
\defedgetype{orderd}{->,line width=0.25mm,blue,densely dotted}{}
\defedgetype{order}{->,line width=0.30mm,blue}{}
\defedgetype{grafo}{->,line width=0.30mm,red}{}
\defedgetype{grafod}{->,line width=0.25mm,red,densely dotted}{}

\def\defn#1{\textbf{#1}}

\def\rR{{\color{red} R} }
\def\bleq{{\color{blue} P}}

\def\F{\mathfrak{F}}

\title{Reasoning about Intuitionistic \\ Computation Tree Logic}
\author{Davide Catta
\institute{Università di Napoli "Federico II"\\ Naples, Italy}
\and
Vadim Malvone
\institute{Télécom Paris\\ Palaiseau, France}
\and
Aniello Murano
\institute{Università di Napoli "Federico II"\\ Naples, Italy}
}
\def\titlerunning{Reasoning about Intuitionistic Computation Tree Logic}
\def\authorrunning{D. Catta, V. Malvone \& A. Murano}
\begin{document}
\maketitle

\begin{abstract}
 In this paper, we define an intuitionistic version of Computation Tree Logic. After explaining the semantic features of intuitionistic logic, we examine how these characteristics can be interesting for formal verification purposes. Subsequently, we define the syntax and semantics of our intuitionistic version of CTL and study some simple properties of the so obtained logic. We conclude by demonstrating that some fixed-point axioms of CTL are not valid in the intuitionistic version of CTL we have defined.
\end{abstract}

\section{Introduction}

Classical modal and temporal logics are extensions of classical logic, in which some new operators (usually called modalities) qualify the truth of classical formulae. For instance, in a classical or temporal modal logic, one can express that a certain formula  is \emph{necessarily} true, \emph{possibly} true,  that \emph{it will be true} in some future moment of time and so on. In particular, temporal logics are a family of modal logics in which the modalities permit to express, as the name suggests, temporal properties of formulae. Temporal logics originated in the philosophical work of Prior~\cite{Prior1955-PRITAM} in the 50s and were rediscovered and adapted by Pnueli~\cite{Pnueli77} who defined the Linear Temporal Logic (LTL) and showed how interesting program properties could be expressed using temporal logic and, more importantly, automatically verified on mathematical models of the executions of such programs. Time flow in LTL is linear, meaning that each instant of time has exactly one successor. 
In 1981, Edmund M. Clarke and E. Allen Emerson first introduced Computation Tree Logic (CTL)~\cite{EmersonC82}. CTL is a type of branching-time logic, where time is represented as a tree-like structure with an undetermined future. This logic was first used to reason about abstractions of concurrent programs and subsequently became a milestone in the automatic verification of cyber-physical models. 

\paragraph{Intuitionistic Logic.}
Intuitionism is a mathematical school developed in the early 1900s by the Dutch mathematician L.E.J. Brouwer. Intuitionism rejects the idea that the truth value of a mathematical statement is independent of our ability to \emph{know or verify} it. Put differently: an intuitionist believes that the truth conditions of a mathematical statement are its  provability conditions. As a result, intuitionism rejects the validity of the principle of excluded middle. In fact, given any mathematical statement $\varphi$, it is not always possible to have knowledge of (i.e., prove or verify) $\varphi$ or its negation $\neg \varphi$. 
In the 1930s, Heyting developed intuitionistic logic~\cite{Mancosu1997-MANFBT-2}, a logic embodying the underlying principles of intuitionistic reasoning. Intuitionistic logic has found many applications in computer science, e.g., the Curry-Howard isomorphism relating intuitionistic proofs and typed lambda terms~\cite{sorensen} and the constructive type theory of Per Martin-L\"of~\cite{DBLP:books/daglib/0000395}. The semantics of intuitionistic logic was first specified by means of topological spaces, and later, by Saul Kripke in terms of Kripke models~\cite{KRIPKEintu}. This latter semantics is particularly interesting for our means. Kripke's idea is that a model of intuitionistic logic represents the dynamics of an \emph{idealized agent} (or idealized mathematician) that is expanding her knowledge about mathematical statements over time.  In this temporal process, she creates new elements while observing the fundamental facts in her universe. Moving from one moment to the next, she freely decides how to continue her activity, resulting in a partially ordered set of possible stages, known as possible worlds. In this particular interpretation, the truth of an intuitionistic formula $\varphi$ in a given moment of time $w$ depends upon the future of $w$, i.e, upon the moments of time coming after $w$. For instance, to conclude that a formula $\varphi \imp \psi$ is true at a given moment $w$, the agent must be certain that for every future moment $w'$, if there is a proof of $\varphi$ at $w'$, it is always possible to obtain a proof of $\psi$ as well.

\paragraph{Intuitionistic Modal Logic.}
Intuitionistic logic can be extended with modalities in different ways (for an overview see  \cite{simpson:phd}):
while in classical logic axioms involving only $\Box$ provide also description of the behavior of $\Diamond$, for intuitionistic logic this is no more the case since the duality of the two modalities does not hold anymore.
This leads to different approaches.
\emph{Constructive modal logics} consider minimal sets of axioms to guarantee the definition of the behaviors of the $\Box$ and $\Diamond$ modalities.
A second approach, referred to as \emph{intuitionistic modal logic}, considers additional axioms in order to validate the G\"{o}del-Gentzen translation  \cite{das:mar:blog}. This second approach has led to the definition of a class of Kripke models (called birelational models) in which two distinct relations of accessibility are considered: one representing the aforementioned preorder and the other representing the “standard" accessibility relation of a Kripke model. In this paper, for the sake of simplicity, we will follow this second approach. This approach is used, for instance, for the intuitionistic version of LTL studied in~\cite{DBLP:journals/tocl/BalbianiBDF20}.

\paragraph{Intuitionistic Computational Tree Logic.} In this paper, we aim to define an intuitionistic version of the aforementioned Computation Tree Logic. The intuition we seek to formalize is as follows: we distinguish between two different temporal evolutions within a CTL model. One represents the agent's knowledge about the system, while the other represents the possible evolutions of the system itself based on the given knowledge. The agent's goal is to conclusively verify that certain properties hold with respect to the possible evolutions of the model and the potential evolution of its knowledge. In other words, given a property of interest $\varphi$, it wants to ascertain that, in any state of its knowledge regarding the model, and regardless of the evolution of the model itself, $\varphi$ is satisfied. We think this type of intuition allows for considering the verification of CTL properties in the imperfect information context.

\paragraph{Structure of the work} In Section~\ref{ICTL} we provide the syntax and semantic of Intuitionistic Computation Tree Logic. Then, in Section~\ref{properties} we provide some properties of our logic. Finally, we conclude in Section~\ref{conclusions} with some future directions.

\section{Syntax and Semantics of Intuitionistic Computation Tree Logic}\label{ICTL}
In this section we provide the syntax and semantics of our logic. We follow~\cite{DBLP:conf/tark/PlotkinS86} for the definition of models that will be used in the following. 
We fix a countable set $\Ap$ of \emph{atomic propositions} or \emph{atoms}. 

\begin{defi}
Formulae of Intuitionistic Computation Tree Logic (ICTL for short) are defined by the following grammar: 
$$\varphi,\psi := p \mid \bot \mid \varphi \land \psi \mid \varphi \vee \psi \mid \varphi \imp \psi $$
$$ \some \next \varphi \mid \some (\varphi \until \psi) \mid \some (\varphi \release \psi)\mid \all \next \varphi \mid \all (\varphi \until \psi) \mid \all (\varphi \release \psi) $$

\noindent where $p\in \Ap$ and $\bot$ is the absurdity symbol. We define the negation of a formula $\varphi$ as $\neg \varphi \equiv \varphi \imp \bot$ .  Formulae whose first operator is $\some$ are called \emph{existential} formulae, while those whose first operator is $\all$ are called \emph{universal} formulae.  
\end{defi}

Given the syntax of ICTL, we can now provide the definition of birelational frame.

\begin{defi}
    A \defn{birelational frame} $\F$ is a triple $\tuple{W,\bleq, \rR}$ where $W$ is a non-empty countable set of worlds, $\bleq$ is a preorder on $W $  (i.e., a reflexive and transitive relation) and $\rR$ is a binary serial relation on $W$ (i.e., for every $x\in W$ there is a $y\in W$ such that $x \rR y$) in which the following conditions are satisfied, for all $x,y,z\in W$:  
   \begin{description}
    \item[($\mathsf{C}_1$)]  if $x \rR y$ and $y\bleq z$, then there is a $u\in W$ such that $x\bleq u$ and $u \rR z$ (see below left); 
    
    \item[($\mathsf{C}_2$)]  if $x\bleq z$ and $x\rR y$, then there is a $u\in W$ such that $y\bleq u$ and $z \rR u$ (see below right).  
\end{description}

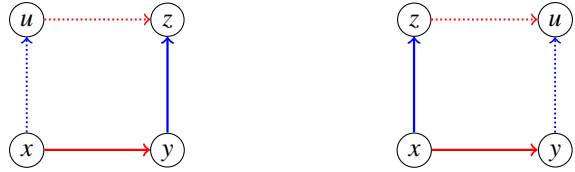
\small 
$$\begin{array}{c@{\qquad\qquad}c@{\qquad\qquad\qquad\qquad}c@{\qquad\qquad}c}
   \vnode{11}{u} & \vnode{12}{z} & \vnode{13}{z} & \vnode{14}{u}
   \\
   \\
   \\
   \\
   \vnode{21}{x} & \vnode{22}{y} &  \vnode{23}{x} & \vnode{24}{y}
   \grafoedges{node21/node22/0, node23/node24/0}
   \orderedges{node22/node12/0,node23/node13/0}
   \orderdedges{node21/node11/0, node24/node14/0}
   \grafodedges{node11/node12/0,node13/node14/0}
\end{array}$$
\normalsize

\end{defi}

A \emph{path} in a frame $\F$ is an infinite sequence of worlds $\rho=\rho_0,\rho_1,\ldots$ such that for any $i\in \mathbb N$ we have that $\rho_i \rR \rho_{i+1}$. If $\rho$ is a path then $\rho_i$ denotes its $(i+1)$-th element, $\rho_{\leq i}$ the finite prefix $\rho_0,\ldots,\rho_i$ of $\rho$ and $\rho_{\geq i}$ the infinite suffix $\rho_i,\rho_{i+1},\ldots$ of $\rho$ starting at $\rho_i$. 

\begin{lemma}\label{lemma:1}
    Let $\F$ be a frame and $w$ and $w'$ two worlds of $\F$ such that $w\bleq w'$. For every path $\rho$ such that $\rho_0=w$ there is a path $\tau$ such that $\tau_0=w'$ for which holds that $\rho_i\bleq \tau_i$ for every natural number $i$.  
\end{lemma}
\begin{proof}
Let $A=(w_i)_{i\in \mathbb{N}}$ be an enumeration of the worlds of $\F$. Given $\rho$, we define $\tau$ by induction on $\mathbb{N}$. We define $\tau_0=w'$ and for any $i\geq 1$ we let $\tau_i$ be the smallest element in $A$ such that $\tau_{i-1} \rR\tau_{i}$ and $\tau_i \bleq \rho_i$. Remark that $\tau_i$ exists because of condition $\mathsf{C}_2$ above: in fact, suppose that for $j<i$ it holds that $\rho_j \bleq \tau_j$, in particular, this means that $\tau_{i-1}\bleq \rho_{i-1}$. Since $\rho_{i-1}\rR \rho_{i}$, then by $\mathsf{C}_2$ there is (at least one and possibly an infinite countable number of)  $u\in W$ such that  $\tau_{i}\rR u$ and $\rho_i \bleq u$.    
\end{proof}

Given the definition of frames, we are able to define our models.

\begin{defi}
    A \defn{birelational model} (model from now on) is a tuple $\M=\tuple{W,\bleq,\rR,\mathcal{V}}$ where $\tuple{W,\bleq,\rR}$ is a birelational frame and $\mathcal{V}: W\to 2^\mathcal{P}$ is a \defn{valuation function} sending each world $w$ to the subset of atomic propositions that are true at $w$ and that is subject to the \defn{monotonicity condition}, that is: if $w\bleq w'$ then $\mathcal{V}(w)\subseteq \mathcal{V}(w')$. 
\end{defi}

Now, we have all the ingredients to define the semantics of ICTL.

\begin{defi}
The satisfaction relation $\M,w\models \varphi$ between a model $\M$, a world $w$ of $\M$, and an ICTL formula $\varphi$ is inductively defined as follows: 

\begin{itemize}
    \item $\M,w\models p$ iff $p\in \mathcal{V}(w)$; 
    \item $\M,w\models \bot$ never; 
    \item $\M,w\models \psi \land \theta$ iff $\M,w\models \psi$ and $\M,w \models \theta$; 
    \item $\M,w \models \psi \vee \theta$ iff $\M,w\models \psi$ or $\M,w\models \theta$; 
    \item $\M,w\models \psi \imp \theta$ iff for every $w'$ such that $w\bleq w'$ we have that $\M,w'\models \psi $ implies $\M,w'\models \theta$; 
    \item $\M,w\models \some \next \psi$ iff  there is a path $\rho$ whose first element $\rho_0$ is $w$ and $\M,\rho_1\models \varphi$; 
    \item $\M,w\models \some (\psi \until \theta)$ iff there is a path $\rho$ whose first element $\rho_0$ is $w$ and there is a $j\geq 0$ such that $\M,\rho_j \models \theta$ and for all $0\leq i < j$ we have that $\M,\rho_i \models  \psi$; 
    \item $\M,w\models \some (\psi \release \theta)$ iff there is a path $\rho$ whose first element $\rho_0$ is $w$ and either $\M,\rho_i \models \theta$ for all $i\in \mathbb{N}$ or there is a $j\geq 0$ such that $\M,\rho_j \models \psi$ and $\M,\rho_i \models \theta$ for all $0\leq i \leq j$; 
     \item $\M,w\models \all \next \psi$ iff for every $w'$ such that $w\bleq w'$ and for every path $\rho$ whose first element $\rho_0$ is $w'$ we have that $\M,\rho_1\models \psi$; 
    \item $\M,w\models \all (\psi \until \theta)$ iff for every $w'$ such that $w\bleq w'$ and for every  path $\rho$ whose first element $\rho_0$ is $w'$ we have that there is a $j\geq 0$ such that $\M,\rho_j \models \theta$ and for all $0 \leq i < j$ we have that $\M,\rho_i \models  \psi$; 
    \item $\M,w\models \all (\psi \release \theta)$ iff for every $w'$ such that $w\bleq w'$ and for every path $\rho$ whose first element $\rho_0$ is $w'$ we have that either $\M,\rho_i \models \theta$ for all $i\in \mathbb{N}$ or there is a $j\geq 0$ such that $\M,\rho_j \models \psi$ and $\M,\rho_i \models \theta$ for all $0\leq i \leq j$. 
    
\end{itemize}
We write $\M,w\not \models \varphi$ when $w$ does not satisfy $\varphi$. We say that a formula $\varphi$ is \defn{valid in a model} $\M$ iff it is satisfied in every $w\in W$. A formula $\varphi$ is \defn{valid in a frame} $\F=\tuple{W,\bleq,\rR}$ iff it is valid in $\tuple{W,\bleq,\rR,\mathcal{V}}$  for any valuation $\mathcal{V}$. A formula $\varphi$ is \defn{valid} iff it is valid in every frame. 
    
\end{defi}

Remark that, by the above definition, a formula $\neg \varphi = \varphi \imp \bot$ is satisfied at $w$ if for any $w'$ such that $w\bleq w'$ we have that $w'$ satisfies $\varphi $ implies $w'$ satisfies $\bot$. Since $\bot$ is \emph{never} satisfied, this is equivalent to say that \emph{no state} $w'$ bigger than $w$ satisfies $\varphi$.  Given the above definition of satisfaction and validity, it is fairly easy to show that the operators are not dual, e.g., we do not have that $\neg \all \next \neg \varphi \imp \some \next \varphi$. 

\section{Main Properties}\label{properties}
In this section, we show that CTL and ICTL differs in some important properties. First, we show that the satisfaction relation is monotonous with respect to the preorder. This properties, that is shared by all intuitionist modal logics, intuitively says that the agent's knowledge can only grow. 
\begin{prop}\label{prop:preservation}
Let $\M$ be a model, $\varphi$ a formula, and $w$ and $w'$ a pair of worlds of $\M$. if $\M,w\models \varphi$ and $w\bleq w'$ then $\M,w' \models \varphi$
    
\end{prop}
\begin{proof}
    By induction on the structure of $\varphi$ using Lemma~\ref{lemma:1} when we consider existential formulae. 
\end{proof}

Given a model $\M$ and a formula $\varphi$, we write $\sem{\varphi}^\M$ to denote the set of worlds of $\M$ satisfying $\varphi$, that is $\sem{\varphi}^\M=\set{w\in W \mid \M,w\models \varphi}$. Whenever the model $\M$ is contextually given and no confusion can arise, we omit the superscript $\M$.
If $w$ is a world of $\M$, we denote by $w^\uparrow$ the set of worlds that are greater of $w$ with respect to $\bleq$. 
Given $X\subseteq W$, we let $Pre^\exists (X)=\set{w' \mid w'\rR w \text{ for some } w\in X}$ and $Pre^\forall(X)=\set{w' \mid w'\rR w  \text{ implies } w\in X}$. If $Y\subseteq W$ is a set then $Y^\uparrow$ denotes the set of elements of $w$ whose upward closure is in $Y$, that is $Y^\uparrow=\set{w\in W \mid w^\uparrow \subseteq Y}$. If $Y$ is a set $Y^\mathbf{c}$ denotes its complement.  

\begin{prop}
   Given a model $\M$ and a world $w$, we have that: 

   \begin{enumerate}
   \item $\M,w\models \varphi \imp \psi$ iff $w\in (\sem{\varphi}^{\mathbf{c}} \cup \sem{\psi})^\uparrow$
       \item $\M,w\models \some \next \varphi$ iff $w\in Pre^\exists(\sem{\varphi})$
       \item $\M,w\models \all \next \varphi$ iff $w\in (Pre^\forall(\sem{\varphi}))^\uparrow$
   \end{enumerate}
\end{prop}

\begin{proof}
We only prove (1) and (3). 
\begin{enumerate}
    \item For the $(\imp)$-direction we reason by contraposition:  suppose that there is $ w'$ such that $w'\in \sem{\varphi} \cap \sem{\psi}^\mathbf{c}$ and $w\bleq w'$. This means that $\M,w'\models \varphi$ and $\M,w'\not\models \psi $ thus we conclude that $\M,w\not\models \varphi \imp \psi$.  For the converse  direction: suppose that $w\in (\sem{\varphi}^\mathbf{c} \cup \sem{\psi})^\uparrow$. Thus given any $w'$ such that $w\bleq w'$ we have that $w'\in \sem{\varphi}^\mathbf{c}$ or $w'\in \sem{\psi}$. From this fact we deduce that $w'\in \sem{\varphi}$ (that is $w\notin\sem{\varphi}^\mathbf{c}$) implies $w'\in \sem{\psi}$ and we can conclude. 
\end{enumerate}

\begin{enumerate}\setcounter{enumi}{2} 
    \item For the $(\imp)$-direction: suppose that $\M,w \models \all \next \varphi$. By definition, this means that given any $w'$ such that $w\bleq w'$ all successors of $w'$
 are in $\sem{\varphi}$. This proves that $w^\uparrow \subseteq Pre^\forall(\sem{\varphi})$ and thus $w\in (Pre^\forall(\sem{\varphi}))^\uparrow$.  For the converse direction, suppose that $s\in (Pre^\forall(\sem{\varphi}))^\uparrow$. This means given any $w'$ such that $w\bleq w'$ we have that all succesors of $w'$ are in $\sem{\varphi}$. Thus any path starting at the given $w'$ will satisfy $\varphi$ on its second component. We thus deduce that $\M,w\models \all \next \varphi$. 

\end{enumerate} \end{proof}

\begin{prop}\label{1stepunfolding}
Define $\varphi \sse \psi$ as $(\varphi \imp \psi ) \land (\psi \imp \varphi)$. The following formulae are valid:
   \begin{enumerate}
        \item\label{someuntil} $\some (\varphi \until \psi) \sse \psi \vee (\varphi \land  \some \next  \some (\varphi \until \psi))$
       \item\label{somerelease} $\some (\varphi \release \psi) \sse \psi \land (\varphi \vee \some \next \some (\varphi \release \psi))$       
        \item\label{alluntil} $  \psi \vee (\varphi \land  \all \next  \all (\varphi \until \psi)) \imp \all (\varphi \until \psi)$
       \item\label{allrelease} $\psi \land (\varphi \vee \all \next \all (\varphi \release \psi))\imp  \all (\varphi \release \psi) $
         \end{enumerate} 
\end{prop}

\begin{proof}
    We prove~(\ref{somerelease}) and~(\ref{alluntil}). Let $\M$ be any model and $w$ any of its worlds. 

    \begin{enumerate}\setcounter{enumi}{1} 
  
  \item For the $(\imp)$-direction, suppose that $\M,w\models \some (\varphi \until \psi)$ and let $w'$ be a world such that $w\bleq w'$. We must check that $w'\models \psi \vee (\varphi \land \some \next \some (\varphi \until \psi))$. From the fact that $w$ satisfies $\some (\varphi \until \psi)$, we deduce that either $w\in \sem{\psi}$ (in this case we conclude by Proposition~\ref{prop:preservation}), or that $w\in \sem{\varphi}$ and there is a path $\rho$ such that $\rho_{0}= w$ , $\rho_{i} \in \sem{\psi}$ for some $i\geq 1$ and $\rho_j\in \sem{\varphi}$ for all $1\leq j <i $, we thus conclude that $w\in \sem{\psi \vee ( \varphi \land  \some \next \some (\varphi \until \psi)}$ and, again by Proposition~\ref{prop:preservation}, that $w'\in \sem{\psi \vee (\varphi \land (\some \next \some (\varphi \until \psi))}$. 
  For the converse direction: suppose that $w\in \sem{\psi \vee ( \varphi\land \some \next \some (\varphi \until \psi)}$ and let $w'$ be a world such that $w\bleq w'$. Since $w\in \sem{\varphi \land (\psi \vee \some \next \some (\varphi \until \psi))}$ either $w\in \sem{\psi} $ or $w\in \sem{\varphi \land \some\next \some (\varphi \until \psi)}$. In both cases, we deduce that $w\in \sem{\some\varphi \until \psi}$ and we conclude using Proposition~\ref{prop:preservation}.

\end{enumerate}

\begin{enumerate}\setcounter{enumi}{2} 
    \item Suppose that $\M,w\models \psi \vee (\varphi \land \all \next \all (\varphi \until \psi))$ and let $w'$ be a world such that $w\bleq w'$. We must show that $w'\in \sem{\all (\varphi \until \psi)}$. If $w\in \sem{\psi}$ we conclude using Proposition~\ref{prop:preservation} since any path starting at any world bigger than $w$ will immediately satisfy $\psi$ (and thus $(\varphi \until \psi)$). Otherwise,  $w \in \sem{\varphi \land \all \next \all (\varphi \until \psi) )}$ this means that $w\in \sem \varphi$ and given any world $w'$ bigger than $w$, we will have that $w'\rR u$ implies $u\in \sem{\all (\varphi \until \psi)}$. Since $w\bleq w'$,   Proposition~\ref{prop:preservation} allows us to  conclude that $w'\in \sem{\varphi}$ and since given any $u$ such that $w' \rR u$ we have that $u\in \sem{\all (\varphi \until \psi)}$, we conclude that $w'\in \sem{\all (\varphi \until \psi)}$ as we wanted.

\end{enumerate}

\end{proof}

Note that, unlike CTL, in ICTL in Proposition~\ref{1stepunfolding}.\ref{alluntil} and \ref{1stepunfolding}.\ref{allrelease} we have an implication. So, to prove that the other directions do not hold, we provide a counterexample in the following proposition.

\begin{prop}
The two formulae below are \textbf{not valid}

\begin{enumerate}
    \item\label{notalluntil}  $\all( p\until q ) \imp q \vee (p \land (\all \next \all (p \until q)) $
    \item\label{notallrelase}  $\all( q\release p ) \imp p\land (q \vee (\all \next \all (q \until p)) $
\end{enumerate}

\end{prop}
\begin{proof}
   For both formulae, consider the model $\M$ depicted below in which the preorder $\bleq$  is represented by the blue arrows, the relation $\rR$ is represented by the red arrows and the valuation function is specified next to each node.
\vspace{0.5cm}
\small
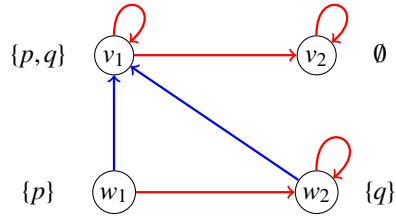
$$\begin{array}{cc@{\qquad\quad}c@{\qquad\quad}cc}
    \vlabel{21}{\set{p,q}} &  \vnode{21}{v_1}  & &\vnode{22}{v_2} & \vlabel{22}{\emptyset}  
      \\
      \\
      \\
      \\
   \vlabel{11}{\set{p}} &   \vnode{11}{w_1} && \vnode{12}{w_2} & \vlabel{12}{\set{q}}
      \orderedges{node11/node21/0,node12/node21/0}
      \specgrafoedge{node12}{node12}{out=90, in=40,loop} 	
       \specgrafoedge{node22}{node22}{out=90, in=40,loop} 
       \specgrafoedge{node21}{node21}{out=90, in=40,loop} 
       \grafoedges{node11/node12/0,node21/node22/0}
     
\end{array}$$
\normalsize
\noindent we have that $\M,w_1 \models\all (p\until q)$ but neither  $\M,w_1 \models q$ nor $\M,w_1 \models p \land \all \next \all (p\until q)$. In particular $w_1$ does not satisfy $\all \next \all (p \until q)$ because  $ w_1 \bleq v_1$ and  given the path $\tau=v_1\cdot v_2^\omega$ we have that there is no $i\geq 1$ such that $\M,\tau_j \models q$. 
 Similarly, $\M,w_1 \models \all (q \release p)$ but $w_1$ does not satisfy neither $p\land q$ nor $p\land \all\next \all (q \release p)$.
\end{proof}

\section{Conclusions and Future Works}\label{conclusions}

We have sketched an Intuitionistic variant of CTL (ICTL) and proved some basic properties about this logic. There is still a lot, practically everything, to be done about this logic. Below, we outline some directions we would like to explore, without any specific hierarchical order.

\myparagraph{Formal verification} 
In addition to its purely theoretical interest, the study of CTL has been fundamental for the development of applications in formal software verification. This is because Kripke models on which this logic is interpreted are particularly suitable for modeling the evolution of reactive systems. In order to make ICTL appealing, we would like to identify a class of reactive systems, or specific problems within these systems, that lend themselves well to being modeled using birelationals models.

\myparagraph{Model Checking} The model checking problem for ICTL is the same as the one of CTL: given a (finite) birelational model $\M$, a formula $\varphi$, and a world $w$ can we decide whether $\M,w\models \varphi$?  This problem is P-space hard. For instance, given a CTL model $\N$ one can see it as a birelational model by setting $u\bleq u'$ iff $u=u'$ for any world of $\N$, and thus reduce the model-checking problem for CTL to the one of ICTL. Furthermore, even though the CTL fix points axioms do not hold in ICTL, we think we can characterize the semantics of the release and until operators by the upward-closure of the usual (classical) fix-point  characterization. 

\myparagraph{Axiomatization \& fix-points}  The fix-point axioms are not valid in the intuitionistic variant of CTL that we have defined. It would be interesting to find out whether an axiomatization of ICTL can be obtained by other means. Possibly, one could think of adding another condition on birelational models in order to validate the axioms. For instance the following, for all $x,y,z\in W$: if $x\bleq y$ and $y\rR z$ then there exist an $u\in W$ such that $x\rR u$ and $u\bleq z$. Diagrammatically this gives: 
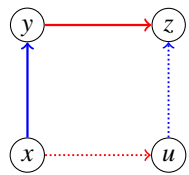
\small 
$$\begin{array}{c@{\qquad\qquad}c}
   \vnode{11}{y} & \vnode{12}{z} 
   \\
   \\
   \\
   \\
   \vnode{21}{x} & \vnode{22}{u} 
   \grafoedges{node11/node12/0}
   \orderedges{node21/node11/0}
   \orderdedges{node22/node12/0}
   \grafodedges{node21/node22/0}
\end{array}$$
\normalsize

\myparagraph{Multiagent Systems} It would be natural to extend our intutionistic intepretation of CTL to  Alternating-time Temporal Logic (ATL)~\cite{DBLP:journals/jacm/AlurHK02}. A natural interpretation of Intuitionistic ATL would be to consider that agents try to verify a temporal formula by executing some coordinate action on the base of a \emph{shared knowledge} at a given time. We think that this interpretation would have much in common with the epistemic interpretation of ATL~\cite{DBLP:journals/sLogica/HoekW03a}.

\bibliographystyle{eptcs}
\bibliography{generic}

\begin{thebibliography}{10}
\providecommand{\bibitemdeclare}[2]{}
\providecommand{\surnamestart}{}
\providecommand{\surnameend}{}
\providecommand{\urlprefix}{Available at }
\providecommand{\url}[1]{\texttt{#1}}
\providecommand{\href}[2]{\texttt{#2}}
\providecommand{\urlalt}[2]{\href{#1}{#2}}
\providecommand{\doi}[1]{doi:\urlalt{https://doi.org/#1}{#1}}
\providecommand{\eprint}[1]{arXiv:\urlalt{https://arxiv.org/abs/#1}{#1}}
\providecommand{\bibinfo}[2]{#2}

\bibitemdeclare{article}{DBLP:journals/jacm/AlurHK02}
\bibitem{DBLP:journals/jacm/AlurHK02}
\bibinfo{author}{Rajeev \surnamestart Alur\surnameend},
  \bibinfo{author}{Thomas~A. \surnamestart Henzinger\surnameend} \&
  \bibinfo{author}{Orna \surnamestart Kupferman\surnameend}
  (\bibinfo{year}{2002}): \emph{\bibinfo{title}{Alternating-time temporal
  logic}}.
\newblock {\slshape \bibinfo{journal}{J. {ACM}}}
  \bibinfo{volume}{49}(\bibinfo{number}{5}), pp. \bibinfo{pages}{672--713},
  \doi{10.1145/585265.585270}.

\bibitemdeclare{article}{DBLP:journals/tocl/BalbianiBDF20}
\bibitem{DBLP:journals/tocl/BalbianiBDF20}
\bibinfo{author}{Philippe \surnamestart Balbiani\surnameend},
  \bibinfo{author}{Joseph \surnamestart Boudou\surnameend},
  \bibinfo{author}{Mart{\'{\i}}n \surnamestart Di{\'{e}}guez\surnameend} \&
  \bibinfo{author}{David \surnamestart Fern{\'{a}}ndez{-}Duque\surnameend}
  (\bibinfo{year}{2020}): \emph{\bibinfo{title}{Intuitionistic Linear Temporal
  Logics}}.
\newblock {\slshape \bibinfo{journal}{{ACM} Trans. Comput. Log.}}
  \bibinfo{volume}{21}(\bibinfo{number}{2}), pp. \bibinfo{pages}{14:1--14:32},
  \doi{10.1145/3365833}.

\bibitemdeclare{misc}{das:mar:blog}
\bibitem{das:mar:blog}
\bibinfo{author}{Anupam \surnamestart Das\surnameend} \& \bibinfo{author}{Sonia
  \surnamestart Marin\surnameend}: \emph{\bibinfo{title}{Brouwer meets Kripke:
  constructivising modal logic}}.
\newblock
  \bibinfo{howpublished}{\url{https://prooftheory.blog/2022/08/19/brouwer-meets-kripke-constructivising-modal-logic/}}.
\newblock \bibinfo{note}{Posted on August 19 2022}.

\bibitemdeclare{article}{EmersonC82}
\bibitem{EmersonC82}
\bibinfo{author}{E.~Allen \surnamestart Emerson\surnameend} \&
  \bibinfo{author}{Edmund~M. \surnamestart Clarke\surnameend}
  (\bibinfo{year}{1982}): \emph{\bibinfo{title}{Using Branching Time Temporal
  Logic to Synthesize Synchronization Skeletons}}.
\newblock {\slshape \bibinfo{journal}{Sci. Comput. Program.}}
  \bibinfo{volume}{2}(\bibinfo{number}{3}), pp. \bibinfo{pages}{241--266},
  \doi{10.1016/0167-6423(83)90017-5}.

\bibitemdeclare{article}{DBLP:journals/sLogica/HoekW03a}
\bibitem{DBLP:journals/sLogica/HoekW03a}
\bibinfo{author}{Wiebe \surnamestart van~der Hoek\surnameend} \&
  \bibinfo{author}{Michael~J. \surnamestart Wooldridge\surnameend}
  (\bibinfo{year}{2003}): \emph{\bibinfo{title}{Cooperation, Knowledge, and
  Time: Alternating-time Temporal Epistemic Logic and its Applications}}.
\newblock {\slshape \bibinfo{journal}{Stud Logica}}
  \bibinfo{volume}{75}(\bibinfo{number}{1}), pp. \bibinfo{pages}{125--157},
  \doi{10.1023/A:1026185103185}.

\bibitemdeclare{incollection}{KRIPKEintu}
\bibitem{KRIPKEintu}
\bibinfo{author}{Saul~A. \surnamestart Kripke\surnameend}
  (\bibinfo{year}{1965}): \emph{\bibinfo{title}{Semantical Analysis of
  Intuitionistic Logic I}}.
\newblock In \bibinfo{editor}{J.N. \surnamestart Crossley\surnameend} \&
  \bibinfo{editor}{M.A.E. \surnamestart Dummett\surnameend}, editors: {\slshape
  \bibinfo{booktitle}{Formal Systems and Recursive Functions}}, {\slshape
  \bibinfo{series}{Studies in Logic and the Foundations of
  Mathematics}}~\bibinfo{volume}{40}, \bibinfo{publisher}{Elsevier}, pp.
  \bibinfo{pages}{92--130}, \doi{10.2307/2270547}.

\bibitemdeclare{book}{Mancosu1997-MANFBT-2}
\bibitem{Mancosu1997-MANFBT-2}
\bibinfo{author}{Paolo \surnamestart Mancosu\surnameend}
  (\bibinfo{year}{1997}): \emph{\bibinfo{title}{From Brouwer to Hilbert: The
  Debate on the Foundations of Mathematics in the 1920S}}.
\newblock \bibinfo{publisher}{Oxford, England: Oxford University Press USA}.

\bibitemdeclare{book}{DBLP:books/daglib/0000395}
\bibitem{DBLP:books/daglib/0000395}
\bibinfo{author}{Per \surnamestart Martin{-}L{\"{o}}f\surnameend}
  (\bibinfo{year}{1984}): \emph{\bibinfo{title}{Intuitionistic type theory}}.
\newblock {\slshape \bibinfo{series}{Studies in proof
  theory}}~\bibinfo{volume}{1}, \bibinfo{publisher}{Bibliopolis}.

\bibitemdeclare{inproceedings}{DBLP:conf/tark/PlotkinS86}
\bibitem{DBLP:conf/tark/PlotkinS86}
\bibinfo{author}{Gordon~D. \surnamestart Plotkin\surnameend} \&
  \bibinfo{author}{Colin \surnamestart Stirling\surnameend}
  (\bibinfo{year}{1986}): \emph{\bibinfo{title}{A Framework for Intuitionistic
  Modal Logics}}.
\newblock In \bibinfo{editor}{Joseph~Y. \surnamestart Halpern\surnameend},
  editor: {\slshape \bibinfo{booktitle}{Proceedings of the 1st Conference on
  Theoretical Aspects of Reasoning about Knowledge, Monterey, CA, USA, March
  1986}}, \bibinfo{publisher}{Morgan Kaufmann}, pp. \bibinfo{pages}{399--406}.

\bibitemdeclare{inproceedings}{Pnueli77}
\bibitem{Pnueli77}
\bibinfo{author}{Amir \surnamestart Pnueli\surnameend} (\bibinfo{year}{1977}):
  \emph{\bibinfo{title}{The Temporal Logic of Programs}}.
\newblock In: {\slshape \bibinfo{booktitle}{18th Annual Symposium on
  Foundations of Computer Science, Providence, Rhode Island, USA, 31 October -
  1 November 1977}}, \bibinfo{publisher}{{IEEE} Computer Society}, pp.
  \bibinfo{pages}{46--57}, \doi{10.1109/SFCS.1977.32}.

\bibitemdeclare{book}{Prior1955-PRITAM}
\bibitem{Prior1955-PRITAM}
\bibinfo{author}{Arthur~N. \surnamestart Prior\surnameend}
  (\bibinfo{year}{1955}): \emph{\bibinfo{title}{Time and Modality}}.
\newblock \bibinfo{publisher}{Westport, Conn.: Greenwood Press}.

\bibitemdeclare{phdthesis}{simpson:phd}
\bibitem{simpson:phd}
\bibinfo{author}{Alex~K. \surnamestart Simpson\surnameend}
  (\bibinfo{year}{1994}): \emph{\bibinfo{title}{The proof theory and semantics
  of intuitionistic modal logic}}.
\newblock Ph.D. thesis, \bibinfo{school}{University of Edinburgh, {UK}}.
\newblock \urlprefix\url{https://hdl.handle.net/1842/407}.

\bibitemdeclare{book}{sorensen}
\bibitem{sorensen}
\bibinfo{author}{Morten~Heine \surnamestart Sørensen\surnameend} \&
  \bibinfo{author}{Pawel \surnamestart Urzyczyn\surnameend}
  (\bibinfo{year}{2006}): \emph{\bibinfo{title}{Lectures on the Curry-Howard
  Isomorphism}}.
\newblock \bibinfo{publisher}{Elsevier}, \doi{10.1016/s0049-237x(06)x8001-1}.
\newblock \urlprefix\url{https://doi.org/10.1016%2Fs0049-237x%2806%29x8001-1}.

\end{thebibliography}
\end{document}